\documentclass{amsart}
\newtheorem{theorem}{Theorem}
\newtheorem{lemma}{Lemma}

\newtheorem{proposition}[theorem]{Proposition}
\newtheorem{example}{Example}
\newtheorem{remark}{Remark}
\newtheorem{definition}{Definition}

\usepackage{graphicx}
\usepackage{multirow}
\usepackage{amsmath,amssymb,amsfonts,amsthm}
\usepackage{mathrsfs}
\usepackage[title]{appendix}
\usepackage{xcolor}
\usepackage{textcomp}
\usepackage{manyfoot}
\usepackage{booktabs}
\usepackage{algorithm}
\usepackage{algorithmicx}
\usepackage{algpseudocode}
\usepackage{listings}
\usepackage{url}
\usepackage{yfonts}
\usepackage{mathrsfs, txfonts, tipa}
\usepackage{pdflscape}
\usepackage{longtable}
\newcommand{\ziv}{\mathbb{Z}_4}

\newcommand{\Z}{\mathbb{Z}_4+\omega\mathbb{Z}_4}
\newcommand{\Zf}{\mathbb{Z}_4[\omega]/\langle \omega^2-\omega \rangle}
\newcommand{\e}{a+\omega b}
\newcommand{\R}{\textfrak{R}}
\newcommand{\la}{\langle}
\newcommand{\ra}{\rangle}
\newcommand{\td}{(\text{\textbaro}, \mathfrak{d})}
\newcommand{\tdl}{(\text{\textbaro}, \mathfrak{d}, \gamma)}
\newcommand{\C}{\textfrak{C}}
\newcommand{\cw}{(\textfrak{c}_0,\textfrak{c}_1,\dots,\textfrak{c}_{n-1})}
\newcommand{\Rx}{\R[x;\text{\textbaro},\mathfrak{d}]}

\begin{document}
	\title{DNA codes from $(\text{\textbaro}, \mathfrak{d}, \gamma)$-constacyclic codes over $\mathbb{Z}_4+\omega\mathbb{Z}_4$}
	\author{Priyanka Sharma}
	\address{Department of Mathematics, Indian Institute of Technology Patna, Patna-801106, India}
	\curraddr{}
	\email{E-mail:priyanka\_2121ma13@iitp.ac.in}
	\thanks{}
	
	\author{Ashutosh Singh}
	\address{Department of Mathematics, Indian Institute of Technology Patna, Patna-801106, India}
	\curraddr{}
	\email{E-mail:ashutosh\_1921ma05@iitp.ac.in}
	\thanks{}
	
	\author{Om Prakash$^{\star}$}
	\address{Department of Mathematics, Indian Institute of Technology Patna, Patna-801106, India}
	\curraddr{}
	\email{om@iitp.ac.in}
	\thanks{* Corresponding author}

	
	
	\dedicatory{}

\begin{abstract}
The article introduces a novel technique for constructing DNA codes from skew cyclic codes over a non-chain extension of $\mathbb{Z}_4.$ We discuss $(\text{\textbaro },\mathfrak{d}, \gamma)$-constacyclic codes over the ring $\textfrak{R}=\mathbb{Z}_4+\omega\mathbb{Z}_4, \omega^2=\omega,$ with an $\textfrak{R}$-automorphism $\text{\textbaro }$ and a $\text{\textbaro }$-derivation $\mathfrak{d}$ over $\textfrak{R}$. By determining the generators of these codes of any arbitrary length over $\R$, we propose a construction on the $(\text{\textbaro },\mathfrak{d},\gamma)$-constacyclic codes to generate additional classical codes with improved and new parameters. Further, we demonstrate the reversibility of these codes and investigate the necessary and sufficient criterion to derive reverse-complement codes. Moreover, we present another construction to generate DNA codes from these reversible codes. Finally, we provide numerous $(\text{\textbaro },\mathfrak{d}, \gamma)$ constacyclic codes and applying the established results, we construct reversible and DNA codes. The parameters of these linear codes over $\mathbb{Z}_4$ are better and optimal according to the available online database of the $\mathbb{Z}_4$ codes.
\end{abstract}
\maketitle

\section{Introduction}
In the mid-$20^{th}$ century, coding theory emerged as a response to the critical need for robust and reliable data transmission systems. The pioneering works of Shannon in the late 1940s laid the groundwork by establishing the mathematical framework for error-correcting codes \cite{Mac}. The initial research was focused on obtaining codes over finite fields. Further, in $1972$, Blake \cite{blake} first studied the codes over certain rings.  After a long gap, in $1994$, the work of Hammons et al. \cite{hammon} attracted researchers to study codes over finite rings \cite{ CodeRing, Shi2017}.
At the same time, a new computing area emerged where Adleman showed a way to solve computationally difficult problems through DNA molecules \cite{adleman}. However, DNA molecules are subject to errors, making error-correcting mechanisms essential in DNA computing \cite{MKGupta}. Researchers have considered combinatorial approaches, algebraic methods, and a few other approaches to solve this problem \cite{Bennenni,guenda,NKumar,Oztas2017, Patel, Yadav}. With the similarity of the number of bases in DNA and elements in $\mathbb{F}_4$ and $\mathbb{Z}_4$, it becomes a suitable choice for the alphabets to construct a DNA code. Reversible code was introduced by Massey in 1964, and it paved the way for DNA code construction \cite{Massey}. Following this, Abualrub and Oehmke \cite{Abualrub1} studied DNA codes over $\mathbb{Z}_4$ of length $2^e$ in 2003. Further, in 2005, Gaborit and King \cite{Gaborit} gave construction for DNA codes and studied the GC-content of the codes. In 2006, Abualrub et al. \cite{Abualrub2} provided a construction for DNA code from cyclic codes over the field $\mathbb{F}_4$. Recently, several works have appeared where cyclic codes over fields and rings were considered for DNA code construction. Further, Bayram et al. \cite{Bayram} explored the codes over $\mathbb{F}_4+v\mathbb{F}_4$ and obtained DNA codes from it in 2016. In 2017, Bennenni et al. \cite{Bennenni} studied DNA cyclic codes over a different ring.
In the quest to get more possibilities for the generators of a code, the study has been extended to the skew polynomial rings \cite{ Boucher07, Boucher09,ore,AS-consta,siap11skew} and their potential in obtaining DNA codes \cite{DNA_256,gursoy,Ashutosh}. This direction is further enriched by the studies of researchers such as Boulagouaz and Leroy \cite{leroy}, Boucher and Ulmer \cite{Boucher14} Sharma and Bhaintwal \cite{Sharma}, Patel and Prakash \cite{shikha2}, and Ashutosh et al. \cite{IJB} where they have taken derivation with associated automorphism to find a generator of the codes.

Here, we adopt an algebraic approach in a more general setup to tackle the error-correcting problem.
This paper first introduces the skew constacyclic codes (SCC codes) with derivations over an extension of $\ziv$. We consider the non-chain ring $\textfrak{R}=\mathbb{Z}_4+\omega\mathbb{Z}_4, \omega^2=\omega$ and explore $\tdl$-constacyclic codes ($\tdl$-CC codes) over this ring.
The novelty of this work is that we investigate the reverse constraint (R-constraint) and reverse-complement constraint (RC-constraint) for a $\tdl$-CC code over some ring. In this way, we obtain several optimal \cite{z4codes} and new linear codes and DNA codes over $\mathbb{Z}_4$ using our proposed theory. Besides the conventional ways, we propose two new constructions to obtain classical and DNA codes. For computational purposes, we use Sagemath \cite{sage} and Magma computer algebra system \cite{magma} to find the factors and Lee distances for the codes.

This article is organized as follows: Section 2 presents the ring structure and subsequent results on its skew polynomial ring. Section 3 contains some basic definitions and results on $\tdl$-CC codes. In Section 4, we explore R- and RC- constraints on the codes obtained in Section 3. Using the automorphism on the ring $\R$ to deal with the reversibility problem, we establish results to obtain R- and RC-codes. We also introduce two new constructions on the codes to generate additional classical codes with optimal and new parameters and DNA codes. In Section 5, we provide examples and tables to highlight the importance of the study. Finally, Section 6 concludes the article by summarizing our results.

\section{Preliminaries}
Throughout the article, $\R$ denotes the ring $\Z= \{\e \mid a,b \in \ziv \}$ where $\omega^2=\omega$. This ring $\R$ has total $9$ ideals, all are principal with two maximal ideals $\la 1+\omega \ra$ and $\la 2+\omega \ra$. Thus, $\R$ is a commutative semi-local, principal ideal ring. It is isomorphic to the quotient ring $\Zf$.
Units of $\R$ are given by $\R^{*} = \{ \e \in \R ~|~ a, a+b \in \ziv^{*} \}$. 
A linear code of length $n$ over the ring $\R$, say $\C$, whose elements are identified as codewords in it, is characterized as an $\R$-submodule of $\R^n$.

The Lee weight over $\mathbb{Z}_4$ is a function that assigns a weight 0 to 0, 1 to 1 and 3, and 2 to 2. 
As an extension, the Lee weight on $\ziv^n$ is calculated by summing the Lee weights of its components. For the generalization of the concept of the Lee weight over $\R$, we use a Gray map $\phi: \R \rightarrow \ziv^{2}$ given by  $$\phi(\e) = (a, a+b)$$
     where $a$ and $b$ are in $\ziv$. Naturally, $\phi $ is $\ziv$-linear distance-preserving which can be generalized to $\Phi : \R^n\rightarrow \ziv^{2n}$ componentwise.
The Lee weight for an element in $\R$ is given by the Lee weight of its $\phi$-image over $\ziv \times \ziv$, i.e., $w_L(\e)= w_L(\phi(\e))$ and thus by extending componentwise, Lee weights in $\R^n$ can be determined. The Lee distance of two distinct elements of $\R^n$ is the Lee weight of their difference. The minimum of the Lee distances of distinct codewords in a code $\C$ is called the Lee distance of $\C$.
  If a code $\C$ has length $n$, size $4^{k_1}2^{k_2}$, $k_1,k_2$ are non-negative integers and the minimum Lee distance $d_L$, then it is represented by $(n,4^{k_1}2^{k_2},d_{L})$. Alternatively, a code $\C$ over $\R$ can be parameterized as $[n,k_1+k_2,d_{L}]$ where $k_1+k_2$ represents the dimension of $\C$.

\begin{definition} Suppose $\R$ is a finite ring with an automorphism $\text{\textbaro }$. An additive map $\mathfrak{d} :\R\to \R$ is called a $\text{\textbaro }$-derivation if $\mathfrak{d}(\textfrak{r}_1\textfrak{r}_2)=\mathfrak{d}(\textfrak{r}_1)\textfrak{r}_2+\text{\textbaro }(\textfrak{r}_1)\mathfrak{d}(\textfrak{r}_2),$ for all $\textfrak{r}_1,\textfrak{r}_2 \in \R.$
\end{definition}

\begin{definition}
Let $\R$ be a ring equipped with an automorphism $\text{\textbaro }$ and a $\text{\textbaro }$-derivation $\mathfrak{d}$. Then 
$$\Rx = \{\textfrak{r}_{0}+\textfrak{r}_{1}x+\cdots+\textfrak{r}_{n}x^{n}\mid \textfrak{r}_{i}\in \R, 0 \leq i \leq n \}$$
forms a ring with the usual polynomial addition and the multiplication of polynomials defined under the rule $x\textfrak{\textfrak{r}}=\text{\textbaro }(\textfrak{r})x+\mathfrak{d}(\textfrak{r})$ for all $\textfrak{r} \in \R$, known as a skew polynomial ring (SPR). 
\end{definition}

\noindent We begin by introducing an automorphism $\text{\textbaro }$ of $\R$ to define the SPR $\R[x;\text{\textbaro },\mathfrak{d}]$. 
Define a map $\text{\textbaro } : \R \rightarrow \R$ as
\begin{align}\label{theta}
    \text{\textbaro }(\e) = a+(1+3\omega)b,
\end{align}
where $a,b \in \ziv.$ Clearly, $\text{\textbaro }$ is an automorphism of order $2$ of $\R$.\\
Again, we take a map $\mathfrak{d} : \R \rightarrow \R$ given as
\begin{align}
   \mathfrak{d}(\e) = \alpha\big( \text{\textbaro }(\e)-(\e) \big)= \alpha (1+2\omega)b
\end{align}
such that $\text{\textbaro }(\alpha)+\alpha=0.$ \\
 Clearly, the map $\mathfrak{d}$ on $\R$, defined as above, is a $\text{\textbaro }$-derivation on it and the nonzero values of $\alpha$ are $2, 1+2\omega$ and $3+2\omega$.
We readily have the following results from the definitions of $\text{\textbaro }$ and $\mathfrak{d} $. 

\begin{lemma}\label{tdcom}
    Let $\textfrak{r}$ be an element in $\R$. Then 
    \begin{enumerate}
        \item $\mathfrak{d}^{n}(\textfrak{r})=0$ for $n>1$.
        \item  $\mathfrak{d}(\text{\textbaro }(\textfrak{r}))+\text{\textbaro }(\mathfrak{d}(\textfrak{r}))=0$.
    \end{enumerate}
\end{lemma}

\noindent Clearly, the multiplication over $\R$ is noncommutative as $x\textfrak{r}=\text{\textbaro }(\textfrak{r})x+\mathfrak{d}(\textfrak{r}) \neq \textfrak{r}x$. Therefore, we generalize it by the induction hypothesis, which is useful in the study of R-codes over $\R.$
\begin{lemma}\label{Gen:xnr}
      Let $\textfrak{r} \in \R$. Then
    $ x^n\textfrak{r}= \begin{cases}
    \textfrak{r}x^n, & \textit{if n is even,}\\
        (\text{\textbaro }(\textfrak{r})x+\mathfrak{d}(\textfrak{r}))x^{n-1}, &\textit{if n is odd.}
        \end{cases} $\\
        In particular, $x^2\textfrak{r}=\textfrak{r}x^2 $ for all $\textfrak{r} \in \R$.
\end{lemma}
\begin{proof}
    Since $x\textfrak{r}=\text{\textbaro }(\textfrak{r})x + \mathfrak{d}(\textfrak{r}),$ by Lemma \ref{tdcom}, we get
    \begin{align}\label{$x^2r$}
        x^2\textfrak{r}=& x(\text{\textbaro }(\textfrak{r})x+\mathfrak{d}(\textfrak{r}))= \text{\textbaro }^2(\textfrak{r})x^2+\mathfrak{d}(\text{\textbaro }(\textfrak{r}))x + \text{\textbaro }(\mathfrak{d}(\textfrak{r}))x + \mathfrak{d}^2(\textfrak{r})
        = \textfrak{r}x^2.
    \end{align}
    Now, if $n$ is even, from successive use of (\ref{$x^2r$}), we get
    $$x^n\textfrak{r}=x^{n-2}(x^2\textfrak{r})= x^{n-2}(\textfrak{r}x^2)= x^{n-4}(x^2\textfrak{r})x^2 = x^{n-4}\textfrak{r}x^4 = \cdots = x^2\textfrak{r} x^{n-2} = \textfrak{r} x^n.$$
If $n$ is odd, $n-1$ is even and then
$$x^n\textfrak{r}=x(x^{n-1}\textfrak{r})=(x\textfrak{r})x^{n-1} = (\text{\textbaro }(\textfrak{r})x+\mathfrak{d}(\textfrak{r}))x^{n-1}.$$
    Hence the result.
\end{proof}
In the ring $\Rx,$ the right division of polynomials is followed by the right division algorithm \cite{Sharma} stated in the next theorem.  
\begin{theorem}\label{divalgo}
Let $\textfrak{f}(x)$ be a polynomial in $\Rx$. Then for a polynomial $\textfrak{h}(x)$ having unit leading coefficient over $\R$, there exist $\textfrak{q}(x)$ and $\textfrak{r}(x)$over $\R$ satisfying
$$\textfrak{f}(x) = \textfrak{q}(x)\textfrak{h}(x)+\textfrak{r}(x)$$
where $\textfrak{r}(x)=0$ or $\deg \textfrak{r}(x)<\deg \textfrak{h}(x)$.
\end{theorem}

\section{$(\text{\textbaro  },\mathfrak{d},\gamma)$-CC codes over $\R$ }
The class of CC codes is an impressive extension of cyclic codes which holds significant importance in coding theory. This section studies SCC codes with derivations over the ring $\R$.

\begin{definition}
    Let $\gamma \in \R^{*}$. Given an automorphism $\text{\textbaro }$ and a $\text{\textbaro }$-derivation $\mathfrak{d} $ of $\R$, a code $\C$ is an SCC code, denoted by $\tdl$-CC code if it is invariant under the $\tdl$-CC shift $\tau_{\tdl}:\R^n \rightarrow \R^n$  given by
$$\tau_{\tdl}\big(\cw\big)= \big(\gamma \text{\textbaro }(\textfrak{c}_{n-1})+\mathfrak{d}(\textfrak{c}_0), \text{\textbaro }(\textfrak{c}_0)+\mathfrak{d}(\textfrak{c}_1), \dots, \text{\textbaro }(\textfrak{c}_{n-2})+\mathfrak{d}(\textfrak{c}_{n-1})\big),$$
for any $\cw \in \C$, i.e., if $\tau_{\tdl}\big(\C\big)=\C$.  Specifically, $\C$ becomes $\td$-cyclic or $\td$-negacyclic code for $\gamma =1$ and $-1$, respectively.
\end{definition}

\noindent We identify $\textfrak{R}_n $ with $\frac{\Rx}{\la \textfrak{f}(x) \ra }$, where $\textfrak{f}(x)$ is a skew polynomial of degree $n$ over $\R$. By associating an element $\textfrak{r}=(\textfrak{r}_0,\textfrak{r}_1,\dots,\textfrak{r}_{n-1})$ to its corresponding skew polynomial $\textfrak{r}(x)=\textfrak{r}_0+\textfrak{r}_1x+\cdots+\textfrak{r}_{n-1}x^{n-1}$, it can be seen that $\frac{\Rx}{\la \textfrak{f}(x) \ra}$ is a left $\Rx$-module under the multiplication given as $\textfrak{a}(x)(\textfrak{r}(x)+\la \textfrak{f}(x) \ra)=\textfrak{a}(x)\textfrak{r}(x)+ \la \textfrak{f}(x) \ra$ for any $\textfrak{a}(x) \in  \Rx$. Thus, throughout this paper, we frequently use the above identification for the codeword $\textfrak{c}$ and its corresponding skew polynomial $\textfrak{c}(x)$ interchangeably.

\begin{lemma}\label{shift}
Let $\textfrak{r}=(\textfrak{r}_0,\textfrak{r}_1,\dots,\textfrak{r}_{n-1}) \in \R^n$ have the skew polynomial representation $\textfrak{r}(x)=\textfrak{r}_0+\textfrak{r}_1x+\textfrak{r}_2x^2+\cdots+\textfrak{r}_{n-1}x^{n-1}$ in $\frac{\Rx}{\la x^n-\gamma \ra }.$ Then $x\textfrak{r}(x)$ represents the element $(\gamma \text{\textbaro }(\textfrak{r}_{n-1})+\mathfrak{d}(\textfrak{r}_0),\text{\textbaro }(\textfrak{r}_0)+\mathfrak{d}(\textfrak{r}_1),\dots,\text{\textbaro }(\textfrak{r}_{n-2})+\mathfrak{d}(\textfrak{r}_{n-1}))$ of $\R^n$.
\end{lemma}
\begin{proof}
It follows a similar approach as given in Proposition 3 in \cite{shikha2}.
\end{proof}

\begin{theorem}
    Suppose $\C$ is a linear code over $\R$ of length $n$. The necessary and sufficient criterion for $\C$ to be $(\text{\textbaro },\mathfrak{d}, \gamma)$-CC is that the skew polynomial correspondence of $\C$ forms a left $\Rx$-submodule of $\frac{\Rx}{\la x^n-\gamma \ra }$.
\end{theorem}
\begin{proof}
   If $\C$ is a linear code over $\R$ of length $n$, $\textfrak{c}_1(x)+\textfrak{c}_2(x) \in \C$, for all $\textfrak{c}_1(x),\textfrak{c}_2(x) \in \C$. Suppose $\C$ is  $(\text{\textbaro },\mathfrak{d},\gamma)$-CC. Then for a codeword $\textfrak{c} \in \C$, its $\tau_{\tdl}$-shift, i.e., $\tau_{\tdl}(\textfrak{c}) \in  \C$. Using Lemma \ref{shift}, we have $x\textfrak{c}(x) \in \C$ for each $\textfrak{c}(x)$ in $\C$.
   Following inductively, $x^i\textfrak{c}(x) \in \C ~\forall i \in \mathbb{N}$. Thus, for every $\textfrak{r}(x) \in \Rx$, we get $\textfrak{r}(x)\textfrak{c}(x) $ in $ \C.$ Thus, $\C$ becomes a left $\Rx$-submodule of $\frac{\Rx}{\la x^n-\gamma \ra }$.\\
   Converse is a straightforward implication from the definition of $\tdl$-CC code and Lemma \ref{shift}.
\end{proof}
Now, we determine the generators of $(\text{\textbaro }, \mathfrak{d}, \gamma)$-CC codes of arbitrary lengths $n$ over $\R.$ If such a code $\C$ which is an $\R$-submodule of $\R^n$, possesses a smallest degree polynomial $\textfrak{g}(x)$ having unit leading coefficient in $\R$, by Theorem \ref{divalgo}, $\C$ can be treated as a principal left $\Rx$-submodule of $\Rx/ \la x^n-\gamma \ra$ generated by $\textfrak{g}(x)$.

\pagebreak
\begin{theorem}
    Let $\C$ be a $\tdl$-CC code over $\R$ of length $n$.
    If $\textfrak{g}(x)$ is a minimum degree polynomial in $\C$ having unit leading coefficient in $\R$, then $\C= \la \textfrak{g}(x) \ra$.
    In addition, $\textfrak{g}(x) |_{r} (x^n- \gamma)$.
\end{theorem}
Consider a right divisor $\textfrak{g}(x)=\sum_{i=0}^{m} \textfrak{g}_i x^i$  of $x^n-\gamma$, then a generator matrix $G$ of the $\tdl$-CC code $\C$ associated with $\textfrak{g}(x)$ is represented as
$$G= \begin{bmatrix}
    \textfrak{g}(x)\\
    x\textfrak{g}(x)\\
    x^2\textfrak{g}(x)\\
        \vdots\\
    x^{n-m-1}\textfrak{g}(x)
\end{bmatrix}.$$
If $ n-m $ is even, then by Lemma \ref{Gen:xnr}, we have
\tiny
$$\text{\normalsize{$G$}} = \begin{bmatrix}
    \textfrak{g}_0 & \textfrak{g}_1 & \textfrak{g}_2 & \dots & \textfrak{g}_m & 0 & 0 & \dots & 0 & 0 \\
    \mathfrak{d}(\textfrak{g}_0) & \text{\textbaro }(\textfrak{g}_0)+ \mathfrak{d}(\textfrak{g}_1) & \text{\textbaro }(\textfrak{g}_1)+ \mathfrak{d}(\textfrak{g}_2) & \dots & \text{\textbaro }(\textfrak{g}_{m-1})+\mathfrak{d}(\textfrak{g}_m) & \text{\textbaro }(\textfrak{g}_m) & 0 & \dots & 0 & 0 \\
    0 & 0 & \textfrak{g}_0 & \dots & \textfrak{g}_{m-2} & \textfrak{g}_{m-1} & \textfrak{g}_m & \dots & 0 & 0\\
    \vdots & \vdots & \vdots & \ddots & \vdots & \vdots & \vdots & \ddots & \vdots & \vdots \\
    0 & 0 & \dots & \dots & \dots & \dots & \dots & \dots & \text{\textbaro }(\textfrak{g}_{m-1})+\mathfrak{d}(\textfrak{g}_m) & \text{\textbaro }(\textfrak{g}_m)
\end{bmatrix}.$$
\normalsize
Similarly, if $ n-m $ is odd, then
\scriptsize
$$\text{\normalsize{$G$}} = {\begin{bmatrix}
    \textfrak{g}_0 & \textfrak{g}_1 & \textfrak{g}_2 & \dots & \textfrak{g}_m & 0 & 0 & \dots & 0 & 0 \\
    \mathfrak{d}(\textfrak{g}_0) & \text{\textbaro }(\textfrak{g}_0)+ \mathfrak{d}(\textfrak{g}_1) & \text{\textbaro }(\textfrak{g}_1)+ \mathfrak{d}(\textfrak{g}_2) & \dots & \text{\textbaro }(\textfrak{g}_{m-1})+\mathfrak{d}(\textfrak{g}_m) & \text{\textbaro }(\textfrak{g}_m) & 0 & \dots & 0 & 0 \\
    0 & 0 & \textfrak{g}_0 & \dots & \textfrak{g}_{m-2} & \textfrak{g}_{m-1} & \textfrak{g}_m & \dots & 0 & 0\\
    \vdots & \vdots & \vdots & \ddots & \vdots & \vdots & \vdots & \ddots & \vdots & \vdots \\
    0 & 0 & \dots & \dots & \dots & \dots & \dots & \dots & \textfrak{g}_{m-1} & \textfrak{g}_m
\end{bmatrix}}.$$
\normalsize
It is evident that there is linear independence between the rows of $G$. Thus, the set $\{ \textfrak{g}(x),x\textfrak{g}(x),$ $x^2\textfrak{g}(x),$ $\dots, x^{n-m-1}\textfrak{g}(x) \}$ constitutes a basis for $\C$ where $m= \deg(\textfrak{g}(x)).$ Thus, we state the subsequent proposition.
\begin{proposition}
    Let $\textfrak{g}(x) |_r x^n-\gamma$ with a unit leading coefficient. Then the $\tdl$-CC code $\C= \la \textfrak{g}(x) \ra$ is a free $\R$-module where $|\C|=|\R^{n-\deg(\textfrak{g}(x))}|.$
\end{proposition}
Now, we introduce an unconventional way to construct linear codes. Using this, we obtain numerous codes over $\ziv$ presented in Tables \ref{tdcyclic} and \ref{CC_Tbl2}.\\

\noindent\textbf{Construction 1:} Let $\textfrak{C}$ be a $\td$-CC code generated by the matrix $G$ over $\textfrak{R}$. Consider $G=M_1+\omega M_2$ where $M_1$ and $M_2$ are matrices over $\mathbb{Z}_4$. Now, we construct a generator matrix $G_M=\begin{bmatrix}
        aM_1&bM_1\\
        cM_2&dM_2
    \end{bmatrix}$, then the code $\textfrak{C}_M$ generated by $G_M$ is a linear code over $\mathbb{Z}_4$. Further, for computational work, we use three different matrices ($\begin{bmatrix}
        a&b\\
        c&d
    \end{bmatrix}=\begin{bmatrix}
        1&0\\
        1&1
    \end{bmatrix}, \begin{bmatrix}
        1&1\\
        0&1
    \end{bmatrix} \text{ and } \begin{bmatrix}
        2&1\\
        3&2
    \end{bmatrix}$ assigned by $N_1, N_2$ and $N_3,$ respectively) to find the codes and their Lee distance.

\begin{algorithm}
\caption{To obtain classical codes by construction 1:}
\begin{algorithmic}
\renewcommand{\algorithmicrequire}{\textbf{Input:}}
\renewcommand{\algorithmicensure}
{\textbf{Output:}}

\Require Ring $\Rx$, matrix $N=\begin{bmatrix}
   a & b \\ c & d
  \end{bmatrix}$, length $n$ and $\gamma$.\\
\textbf{Step 1:} Find the right divisor $g(x)$ of $x^n-\gamma$ in $\Rx$.\\
\textbf{Step 2:} Construct the generator matrix $G$ for the polynomial $g(x)$ over $\R$.\\
\textbf{Step 3:} Find the matrices $M_1$ and $M_2$ over $\ziv$ such that $G=M_1+\omega M_2$.\\
\textbf{Step 4:} Construct a matrix $G_M$ such that $G_M=\begin{bmatrix} aM_1 & bM_1 \\ cM_2 & dM_2 \end{bmatrix}$.\\
\textbf{Step 5:} $\mathfrak{C}:= \text{LinearCode}(G_M);$
\Ensure Parameters($\mathfrak{C}$).
\end{algorithmic}
\end{algorithm}

\section{DNA code construction over $\R$ }
This section presents an algebraic methodology for constructing DNA codes over $\R$. We investigate the R- and RC-constraints for a $\tdl$-CC code $\C$ over $\R$. Note that for the elements of $\ziv$, we use the correspondence with the DNA bases which maps $0$ to $A$, $1$ to $T$, $2$ to $C$ and $3$ to $G$. Here, the DNA nucleotides follow the Watson-Crick complement rule, i.e., the nucleotides $A$ and $T$ are complements of each other, and $C$ and $G$ are complements of each other, symbolically,  $A^c=T$, $T^c=A$, $C^c=G$ and $G^c=C$.
\begin{definition}
Suppose $\textfrak{D}$ is a linear code of length $n$ such that the coordinates of its codewords are from the set of DNA bases $\{A,T,G,C\}$.
We define the primary DNA constraints as follows.
\begin{enumerate}
    \item R-constraint: If for a codeword $x=(x_0,  x_1, \dots, x_{n-1} )$ of $\textfrak{D}$, its reverse given as $x^r=(x_{n-1}, x_{n-2},$ $ \dots, x_0)$ is again a codeword of $\textfrak{D}$, then $\textfrak{D}$ is called an R-code.
    \item RC-constraint: If for a codeword $x=(x_0,  x_1, \dots, x_{n-1} )$ of $\textfrak{D}$, its reverse-complement given as  $x^{rc}= (x_{n-1}^c,\dots, x_{1}^c,x_0^c)$ is again a codeword of $\textfrak{D}$, then $\textfrak{D}$ is called an RC-code.
    \end{enumerate}
    We call $\textfrak{D}$ a DNA code of length $n$ if it satisfies above the constraints.
\end{definition}
Since the order of $\R$ is an exponent of $4$, it is natural to search for a correspondence between the codewords over $\R$ and the sequences of a DNA code. In light of the Gray map $\phi$ which is defined as
    $\phi(\e) \mapsto (a,a+b), $ where $a, b \in \R,$
 we form a one-to-one correspondence shown in Table \ref{DNAcorres.}, between the elements of $\R$ and the DNA pairs. 

\renewcommand{\arraystretch}{1.1}
\begin{table}[ht]
\scriptsize
\centering
    \begin{tabular}{|c|c|c|c|c|c|}
\hline
 Elements of $\R$ & Gray images  & Corresponding & Element of $\R$ & Gray images  & Corresponding \\
    $a+\omega b$ & $\phi(a+\omega b)$  & DNA pairs &   $a+\omega b$ & $\phi(a+\omega b)$  & DNA pairs  \\
    \hline
$0$ & $(0,0)$  & $AA$ & $\omega$  & $(0,1)$  & $AT$ \\
$1$ & $(1,1)$  & $TT$ & $1+\omega$ & $(1,2)$  & $TC$ \\
$2$ & $(2,2)$  & $CC$ & $2+\omega$ & $(2,3)$  & $CG$ \\
$3$ & $(3,3)$  & $GG$ & $3+\omega$ & $(3,0)$  & $GA$ \\
$2\omega$ & $(0,2)$  &  $AC$ & $3\omega$  & $(0,3)$  & $AG$ \\
$1+2\omega$ & $(1,3)$  & $TG$ & $1+3\omega$ & $(1,0)$  & $TA$ \\
$2+2\omega$ & $(2,0)$  & $CA$ & $2+3\omega$ & $(2,1)$  & $CT$ \\
$3+2\omega$ & $(3,1)$  & $GT$ & $3+3\omega$ & $(3,2)$  & $GC$ \\
\hline
    \end{tabular}
    \caption{DNA correspondence with respect to the Gray map $a+\omega b \mapsto (a,a+b)$}
    \label{DNAcorres.}
    \end{table}
\noindent Since we are considering DNA codes with coordinates of every codeword from $\R$, each component in the corresponding DNA sequence is a pair of DNA bases instead of a single DNA alphabet given in Table \ref{DNAcorres.}. Thus, while taking the reverse of a codeword, a reversibility problem is encountered.\\
To understand the reversibility problem, we consider the $\td$-cyclic code $\C= \la x^6 + (2\omega + 2)x^5 + x^4 + x^2 + (2\omega + 2)x + 1\ra$ of length $8$. In the reversible code $\Phi(\C)$, $GCGAGCTAGCGA$ $GCTA$ is a DNA sequence corresponding to the codeword $(3+3\omega)+(3+\omega)x+(3+3\omega)x^2+(1+3\omega)x^3+(3+3\omega)x^4+(3+\omega)x^5+(3+3\omega)x^6+(1+3\omega)x^7 \in \C$. The reverse of this codeword over $\R$ is $ (1+3\omega)+(3+3\omega)x+(3+\omega)x^2+(3+3\omega)x^3+(1+3\omega)x^4+(3+3\omega)x^5+(3+\omega)x^6+(3+3\omega)x^7$ which corresponds to the DNA sequence $TAGCGAGCTAGCGAGC$ but the reverse of $GCGAGCTAGCGAGCTA$ is $ATCGAGCGATCGAGCG$, i.e., the reversibility of codewords of a code over $\R$ does not assure the reversibility of the DNA sequences of the corresponding DNA code. To resolve the reversibility problem that occurs due to DNA pairs, we present a technique that contemplates the reverse of the Gray image of the codeword over $\ziv$.\\
Here, we incorporate the automorphism used to define the skewness of the ring, which we found an efficient way to resolve the reversibility problem.
We see that the reverse of the DNA correspondence of a codeword $\textfrak{c}=(\textfrak{c}_0,\textfrak{c}_1,\dots,\textfrak{c}_{n-1})$ with $\textfrak{c}_i$'s are in $\R$ corresponds to $(\text{\textbaro }(\textfrak{c}_{n-1}),\dots,\text{\textbaro }(\textfrak{c}_1),\text{\textbaro }(\textfrak{c}_0))$. We further denote it by $\textfrak{c}^R=(\textfrak{c}_{n-1}^r,\dots,\textfrak{c}_{1}^r,\textfrak{c}_{0}^r)$, which represents the reverse of the Gray image of $\textfrak{c}$ over $\ziv$. In a similar manner, the reverse-complement of $\textfrak{c}$ is given as $\textfrak{c}^{RC}=(\textfrak{c}_{n-1}^{rc},\dots,\textfrak{c}_{1}^{rc},\textfrak{c}_{0}^{rc})$.\\
To obtain the R-constraint, we use the notion of
 self-reciprocal polynomial. A polynomial $f(x)= \sum_{i=0}^{m} f_i x^i$ over $\ziv$ is self-reciprocal if $f^{*}(x) = f(x)$, where $f^{*}(x)= \sum_{i=0}^{m} f_i x^{m-i}$ is the reciprocal polynomial of $f(x)$.
\begin{theorem}\label{CREven}
    Let $\C$ be a $\tdl$-CC code over $\R$ of even length $n$ such that $\C =\langle \textfrak{g}(x) \rangle  = \langle g_1(x)+\omega g_2(x) \rangle $,
     where $\textfrak{g}(x)$ is a monic polynomial such that $\textfrak{g}(x) |_{r} (x^n-\gamma)$ and $g_1(x), g_2(x) \in \ziv[x]$. Then $\C$ is an R-code if
     \begin{itemize}
         \item[1.] $g_1(x) = g_1^{*}(x)$;
         \item[2.] $g^{*}_2(x)x^{m-m'} = g_2(x)$;
         \item[3.] $\begin{cases}
                     \mathfrak{d}(\omega)g_2(x) \in \C ,& \textit{if $m$ is even}, \\
                    (1+2\omega)g_2(x) \in \C ,& \textit{if $m$ is odd};  \end{cases}$
     \end{itemize}
     where degrees of $g_1(x)$ and $g_2(x)$ are $m$ and $m'$ with $m>m'$, respectively.
\end{theorem}
\begin{proof}
 Let $\C  =\langle \textfrak{g}(x) \rangle$ be a $\tdl$-CC code over $\R$ of even length $n$. Assume that the conditions hold. Let $\textfrak{c}(x)$ be the skew polynomial representation of a codeword $\textfrak{c}$ of $\C$. Then we have a polynomial $\textfrak{a}(x)$ whose degree is at most $k-1-s=n-m-1-s$ with $0 \leq s \leq k-1$ over $\R$ such that
\allowdisplaybreaks
    \begin{align*}
         \textfrak{c}(x) = & \textfrak{a}(x)\textfrak{g}(x)\\ =& \textfrak{a}(x)\big( g_1(x)+\omega g_2(x) \big)\\
              = & \sum_{i=0}^{k-1-s} \textfrak{a}_ix^i \bigg[  \sum_{j=0}^{m} g_{1j}x^{j}+ \omega \sum_{j=0}^{m'} g_{2j}x^{j}\bigg]\\
              \text{or }  & \sum_{i=0}^{k-1} \textfrak{a}_ix^i \bigg[  \sum_{j=0}^{m} g_{1j}x^{j}+ \omega \sum_{j=0}^{m'} g_{2j}x^{j}\bigg],
    \end{align*}
 \textbf{Case I:} When $\deg(g_1(x)) = m$ is even:
\begin{align*}
        \textfrak{c}(x) = & \sum_{i=0}^{(k-2)/2} \textfrak{a}_{(2i)}x^{2i} \bigg[  \sum_{j=0}^{m} g_{1j}x^{j}+ \omega \sum_{j=0}^{m'} g_{2j}x^{j}\bigg]  + \sum_{i=0}^{(k-2)/2} \textfrak{a}_{(2i+1)}x^{2i+1} \bigg[  \sum_{j=0}^{m} g_{1j}x^{j} \\ & + \omega \sum_{j=0}^{m'} g_{2j}x^{j}\bigg] \\
                = & \sum_{i=0}^{(k-2)/2} \textfrak{a}_{(2i)}\bigg[  \sum_{j=0}^{m} g_{1j}x^{2i+j}+ \omega \sum_{j=0}^{m'} g_{2j}x^{2i+j} \bigg]
                 + \sum_{i=0}^{(k-2)/2} \textfrak{a}_{(2i+1)}  \bigg[  \sum_{j=0}^{m} g_{1j}x^{2i+j+1} \\ & + \sum_{j=0}^{m'} \big(\text{\textbaro }(\omega g_{2j})x+\mathfrak{d}(\omega g_{2j})\big)x^{2i+j}\bigg] \\
                = & \sum_{i=0}^{(k -2)/2} \textfrak{a}_{(2i)}\bigg[  \sum_{j=0}^{m} g_{1j}x^{2i+j}+ \omega \sum_{j=0}^{m'} g_{2j}x^{2i+j}\bigg]
                  + \sum_{i=0}^{(k -2)/2} \textfrak{a}_{(2i+1)}  \bigg[  \sum_{j=0}^{m} g_{1j}x^{2i+j+1}\\ & +  \text{\textbaro }(\omega) \sum_{j=0}^{m'} g_{2j}x^{2i+j+1} + \mathfrak{d}(\omega) \sum_{j=0}^{m'} g_{2j}x^{2i+j}
                  \bigg].
     \end{align*}
The reverse of the codeword $\textfrak{c}$ corresponds to the skew polynomial $\textfrak{c}(x)$ over $\ziv$ is given by
\begin{align*}
      & \textfrak{c}^{R}(x)\\  = & \sum_{i=0}^{(k -2)/2} \textfrak{a}_{(2i)}\bigg[  \sum_{j=0}^{m} g_{1j}x^{n-1-(2i+j)}+ (1+3\omega) \sum_{j=0}^{m'} g_{2j}x^{n-1-(2i+j)}\bigg]
                  \\ &+ \sum_{i=0}^{(k -2)/2} \textfrak{a}_{(2i+1)}   \bigg[  \sum_{j=0}^{m} g_{1j}x^{n-1-(2i+j+1)}+ \omega \sum_{j=0}^{m'} g_{2j}x^{n-1-(2i+j+1)}  +  3\mathfrak{d}(\omega) \sum_{j=0}^{m'}g_{2j}x^{n-1-(2i+j)}
                  \bigg]\\
                  = & \sum_{i=0}^{(k -2)/2} \textfrak{a}_{(2i)}\bigg[ \bigg( \sum_{j=0}^{m} g_{1j}x^{m-j} \bigg)x^{k -1-2i}+ (1+3\omega)  \bigg( \sum_{j=0}^{m'} g_{2j}x^{m'-j}\bigg)  x^{(m-m')+(k-1-2i)} \bigg] \\&+ \sum_{i=0}^{(k-2)/2} \textfrak{a}_{(2i+1)}
                   \bigg[ \bigg( \sum_{j=0}^{m} g_{1j}x^{m-j}\bigg)x^{k-2-2i} +\omega  \bigg( \sum_{j=0}^{m'} g_{2j}x^{m'-j}\bigg)x^{(m-m')+(k-2-2i)} \\&
                   + 3\mathfrak{d}(\omega) \bigg( \sum_{j=0}^{m'} g_{2j}x^{m'-j}\bigg)x^{(m-m')+(k-1-2i)} \bigg] \\
                 = & \sum_{i=0}^{(k -2)/2} \textfrak{a}_{(2i)}\bigg[ g^{*}_1(x) x^{k-1-2i}+ (1+3\omega)  \Big( g_2^{*}(x) x^{m-m'}\Big)x^{k-1-2i}\bigg] \\&
                 + \sum_{i=0}^{(k -2)/2} \textfrak{a}_{(2i+1)}
                   \bigg[  g^{*}_1(x)x^{k-2-2i} + \omega  \Big( g_2^{*}(x) x^{m-m'}\Big)x^{k-2-2i} + 3\mathfrak{d}(\omega) \Big( g_2^{*}(x) x^{m-m'}\Big)  x^{k -1-2i} \bigg].
     \end{align*}
   As $g_1(x) = g_1^{*}(x)$ and $g^{*}_2(x)x^{m-m'} = g_2(x)$, we get
     \begin{align*}
     & \textfrak{c}^{R}(x) \\  = & \sum_{i=0}^{(k -2)/2} \textfrak{a}_{(2i)}\bigg[ g_1(x) x^{k-1-2i}+ (1+3\omega)   g_2(x) x^{k-1-2i}\bigg]
                 + \sum_{i=0}^{(k-2)/2} \textfrak{a}_{(2i+1)}
                   \bigg[  g_1(x)x^{k-2-2i}  \\ &  + \omega  g_2(x) x^{k-2-2i} + 3\mathfrak{d}(\omega)  g_2(x)  x^{k -1-2i} \bigg] \\
                  = & \sum_{i=0}^{(k-2)/2} \textfrak{a}_{(2i)} x^{k-2-2i} \bigg[  g_1(x) x + (1+3\omega)  g_2(x)x \bigg] + \sum_{i=0}^{(k-2)/2} \textfrak{a}_{(2i+1)} x^{k-2-2i}
                   \bigg[  g_1(x) + \omega  g_2(x) \\& + 3\mathfrak{d}(\omega) g_2(x) x  \bigg] \\
                  = & \sum_{i=0}^{(k -2)/2} \textfrak{a}_{(2i)} x^{k -2-2i} \bigg[ xg_1(x) + x\text{\textbaro }(1+3\omega)  g_2(x) +\mathfrak{d}(1+3\omega)  g_2(x)\bigg] \\&  + \sum_{i=0}^{(k -2)/2} \textfrak{a}_{(2i+1)} x^{k -2-2i}
                  \bigg[  g_1(x)  + \omega  g_2(x) + x \big(3\mathfrak{d}(\omega)\big)  g_2(x)  \bigg] \\
                   = & \bigg[ \sum_{i=0}^{(k -2)/2} \textfrak{a}_{(2i)} x^{k-1-2i} + \sum_{i=0}^{(k -2)/2} \textfrak{a}_{(2i+1)} x^{k -2-2i} \bigg]  \Big( g_1(x) + \omega  g_2(x) \Big) \\& + \bigg[ \sum_{i=0}^{(k -2)/2} \textfrak{a}_{(2i)} x^{k -2-2i} + \sum_{i=0}^{(k -2)/2} \textfrak{a}_{(2i+1)} x^{k -1-2i} \bigg]
                   \Big( 3\mathfrak{d}(\omega) g_2(x) \Big) \\
                    = & \sum_{i=0}^{k-1} \textfrak{a}_{i} x^{k-1-i} \Big( g_1(x) + \omega  g_2(x) \Big)+ \sum_{i=0}^{(k -2)/2} \big( \textfrak{a}_{(2i)}
                   +  \textfrak{a}_{(2i+1)} x \big) x^{k -2-2i}\Big( 3\mathfrak{d}(\omega) g_2(x) \Big)\\
                   = & \textfrak{a}^{*}(x) x^s \Big( g_1(x) + \omega  g_2(x) \Big)+ \sum_{i=0}^{(k -2)/2} \big( \textfrak{a}_{(2i)}
                   +  \textfrak{a}_{(2i+1)} x \big) x^{k -2-2i}\Big( 3\mathfrak{d}(\omega) g_2(x) \Big).
     \end{align*}
     If $ \mathfrak{d}(\omega)g_2(x) \in \C$, $\textfrak{c}^{R}(x) \in  \C.$ Thus, we conclude that $\C$ is an R-code.\\
\noindent \textbf{Case II:} When $\deg(g_1(x)) = m$ is odd:
\begin{align*}
        \textfrak{c}(x) = & \sum_{i=0}^{(k-1)/2} \textfrak{a}_{(2i)}x^{2i} \bigg[  \sum_{j=0}^{m} g_{1j}x^{j}+ \omega \sum_{j=0}^{m'} g_{2j}x^{j}\bigg]
                  + \sum_{i=0}^{(k-3)/2} \textfrak{a}_{(2i+1)}x^{2i+1} \bigg[  \sum_{j=0}^{m} g_{1j}x^{j} \\& + \omega \sum_{j=0}^{m'} g_{2j}x^{j}\bigg] \\
                = & \sum_{i=0}^{(k-1)/2} \textfrak{a}_{(2i)}\bigg[  \sum_{j=0}^{m} g_{1j}x^{2i+j}+ \omega \sum_{j=0}^{m'} g_{2j}x^{2i+j}\bigg]
                 + \sum_{i=0}^{(k-3)/2} \textfrak{a}_{(2i+1)}  \bigg[  \sum_{j=0}^{m} g_{1j}x^{2i+j+1} \\ & + \text{\textbaro }(\omega) \sum_{j=0}^{m'} g_{2j}x^{2i+j+1}+\mathfrak{d}(\omega) \sum_{j=0}^{m'} g_{2j}x^{2i+j}\bigg].
     \end{align*}
The reverse of the codeword $\textfrak{c}$ corresponds to the skew polynomial $\textfrak{c}(x)$ over $\ziv$ is given by
      \begin{align*}
     & \textfrak{c}^{R}(x)\\ = & \sum_{i=0}^{(k-1)/2} \textfrak{a}_{(2i)}\bigg[  \sum_{j=0}^{m} g_{1j}x^{n-1-(2i+j)}+ (1+3\omega) \sum_{j=0}^{m'} g_{2j}x^{n-1-(2i+j)}\bigg] \\&
                 + \sum_{i=0}^{(k-3)/2} \textfrak{a}_{(2i+1)}  \bigg[  \sum_{j=0}^{m} g_{1j}x^{n-1-(2i+j+1)}  + \omega  \sum_{j=0}^{m'} g_{2j}x^{n-1-(2i+j+1)} +3\mathfrak{d}(\omega) \sum_{j=0}^{m'} g_{2j}x^{n-1-(2i+j)}\bigg]  \\
                  = & \sum_{i=0}^{(k-1)/2} \textfrak{a}_{(2i)}\bigg[ \bigg( \sum_{j=0}^{m} g_{1j}x^{m-j} \bigg)x^{k-1-2i}+ (1+3\omega)  \bigg( \sum_{j=0}^{m'} g_{2j}x^{m'-j}\bigg)x^{(m-m')+(k-1-2i)} \bigg] \\
                  & + \sum_{i=0}^{(k-3)/2} \textfrak{a}_{(2i+1)}
                   \bigg[ \bigg( \sum_{j=0}^{m} g_{1j}x^{m-j}\bigg) x^{k-2-2i} + \omega  \bigg( \sum_{j=0}^{m'} g_{2j}x^{m'-j}\bigg)x^{(m-m')+(k-2-2i)}\\&
                   + 3\mathfrak{d}(\omega) \bigg( \sum_{j=0}^{m'} g_{2j}x^{m'-j}\bigg)
                 x^{(m-m')+(k-1-2i)}\bigg] \\
                 = & \sum_{i=0}^{(k-1)/2} \textfrak{a}_{(2i)}\bigg[ g^{*}_1(x) x^{k-1-2i}+ (1+3\omega)  \Big( g_2^{*}(x) x^{m-m'}\Big)x^{k-1-2i}\bigg]  \\&
                 + \sum_{i=0}^{(k-3)/2} \textfrak{a}_{(2i+1)}
                   \bigg[  g^{*}_1(x)x^{k-2-2i}  + \omega  \Big( g_2^{*}(x) x^{m-m'}\Big)x^{k-2-2i}   + 3\mathfrak{d}(\omega) \Big( g_2^{*}(x) x^{m-m'}\Big)x^{k-1-2i}\bigg].
     \end{align*}
As $g_1(x) = g_1^{*}(x)$ and $g^{*}_2(x)x^{m-m'} = g_2(x)$, we get
     \begin{align*}
      &  \textfrak{c}^{R}(x) \\ = & \sum_{i=0}^{(k-1)/2} \textfrak{a}_{(2i)}\bigg[ g_1(x) x^{k-1-2i}+ (1+3\omega)   g_2(x)x^{k-1-2i}\bigg]  \\&
                 + \sum_{i=0}^{(k-3)/2} \textfrak{a}_{(2i+1)}
                   \bigg[  g_1(x)x^{k-2-2i} + \omega  g_2(x) x^{k-2-2i}  + 3\mathfrak{d}(\omega)  g_2(x) x^{k-1-2i}\bigg] \\
                  = & \sum_{i=0}^{(k-1)/2} \textfrak{a}_{(2i)} x^{k-1-2i} \bigg[ g_1(x) + (1+3\omega)  g_2(x) \bigg] + \sum_{i=0}^{(k-3)/2} \textfrak{a}_{(2i+1)} x^{k-3-2i}
                   \bigg[  g_1(x)x \\ & + \omega  g_2(x)x \bigg]
                 + \sum_{i=0}^{(k-3)/2} \textfrak{a}_{(2i+1)} x^{k-1-2i}
                   \bigg[ 3\mathfrak{d}(\omega) g_2(x)\bigg] \\
                  = & \sum_{i=0}^{(k-1)/2} \textfrak{a}_{(2i)} x^{k-1-2i} \bigg[ g_1(x) + (1+3\omega)  g_2(x)  \bigg] + \sum_{i=0}^{(k-3)/2} \textfrak{a}_{(2i+1)} x^{k-3-2i}
                   \bigg[ x g_1(x) \\& + x \text{\textbaro }(\omega) g_2(x)  + \mathfrak{d}(\omega) g_2(x)\bigg]+ \sum_{i=0}^{(k-3)/2} \textfrak{a}_{(2i+1)} x^{k-1-2i}
                   \bigg[ 3\mathfrak{d}(\omega) g_2(x) \bigg] \\
                   = & \sum_{i=0}^{(k-1)/2} \textfrak{a}_{(2i)} x^{k-1-2i} \bigg[ g_1(x) + (1+3\omega)  g_2(x) \bigg] + \sum_{i=0}^{(k-3)/2} \textfrak{a}_{(2i+1)} x^{k-2-2i}
                   \bigg[  g_1(x)  + (1+3\omega)  g_2(x) \bigg]  \\& + \sum_{i=0}^{(k-3)/2} \textfrak{a}_{(2i+1)} x^{k-3-2i}
                   \bigg[ \mathfrak{d}(\omega) g_2(x) \bigg]  + \sum_{i=0}^{(k-3)/2} \textfrak{a}_{(2i+1)}  x^{k-1-2i}
                   \bigg[ 3\mathfrak{d}(\omega) g_2(x)\bigg] \\
                    = & \sum_{i=0}^{k-1} \textfrak{a}_i x^{k-1-i} \bigg[ g_1(x) + \omega  g_2(x) \bigg] + \sum_{i=0}^{k-1} \textfrak{a}_i x^{k-1-i}
                 \big( (1+2\omega)  g_2(x) \big) \\& + \sum_{i=0}^{(k-3)/2} \textfrak{a}_{(2i+1)} x^{k-3-2i}  (3x^2+1)
                   \bigg( \mathfrak{d}(\omega) g_2(x) \bigg)\\
                    = & \textfrak{a}^{*}(x) x^{s} \Big( g_1(x) + \omega  g_2(x) \Big)
                   + a^{*}(x) x^{s} \Big( (1+2\omega)  g_2(x) \Big) \\& + \sum_{i=0}^{(k-3)/2} \textfrak{a}_{(2i+1)} x^{k-3-2i}  (3x^2+1)
                   \Big( \mathfrak{d}(\omega) g_2(x) \Big).
     \end{align*}
     If $(1+2\omega) g_2(x) \in \C$, $\textfrak{c}^{R}(x) \in  \C.$ Thus, we conclude that $\C$ is an R-code.
\end{proof}
\noindent Analogously, we state the reversibility constraint for a $\tdl$-CC code over $\R$ for odd lengths.
\begin{theorem}\label{CROdd}
    Let $\C$ be a $\tdl$-CC code over $\R$ of odd length $n$ such that $\C =\langle \textfrak{g}(x) \rangle =\la g_1(x)+\omega g_2(x) \ra $,
    where $\textfrak{g}(x)$ is a monic polynomial such that $\textfrak{g}(x) |_{r} (x^n-\gamma)$ and $g_1(x), g_2(x) \in \ziv[x]$. Then $\C$ is an R-code if
     \begin{itemize}
         \item[1.] $g_1(x) = g_1^{*}(x)$;
         \item[2.] $g^{*}_2(x)x^{m-m'} = g_2(x)$;
         \item[3.] $\begin{cases}
                    (1+2\omega) g_2(x) \in \C ,& \textit{if $m$ is even}, \\   \mathfrak{d}(\omega)g_2(x) \in \C ,& \textit{if $m$ is odd};
                    \end{cases}$
     \end{itemize}
     where degrees of $g_1(x)$ and $g_2(x)$ are $m$ and $m'$ with $m>m'$, respectively.
\end{theorem}
\begin{proof}
The proof proceeds in the same manner as Theorem \ref{CREven}.
\end{proof}
\noindent Using the Gray map, we give some useful relations between the elements of $\R$ and their complements in the next lemma, which are readily implied by Table \ref{DNAcorres.}.
\begin{lemma}
    For every $\textfrak{a}, \textfrak{b} \in \R,$ we get
    \begin{enumerate}
        \item $\textfrak{a}+\textfrak{a}^{c} =1.$
        \item $\textfrak{a}^{c}+\textfrak{b}^{c}=(\textfrak{a}+\textfrak{b})^{c}+1.$
        \item $(\textfrak{a}+\textfrak{b})^r = \text{\textbaro }(\textfrak{a}) + \text{\textbaro }(\textfrak{b}).$
    \end{enumerate}
\end{lemma}
Now, we derive a necessary and sufficient criterion for a $\tdl$-CC code to be an RC-code.
\begin{theorem}\label{DNA}
    Let $\C$ be a $\tdl$-CC code over $\R$ of length $n$. Then $\C$ is an RC-code if and only if $\C$ is an R-code and $1+x+\cdots+x^{n-1} \in \C$.
\end{theorem}
\begin{proof}
    Let $\C$ be an $n$-length $\tdl$-CC code over $\R$.
    Suppose  $\C$ is an R-code and $1+x+\cdots+x^{n-1} \in \C$. If  $\textfrak{c}=(\textfrak{c}_0,\textfrak{c}_1,\dots,\textfrak{c}_{n-1})$ is a codeword in $\C,$ i.e., the corresponding polynomial representation  $\textfrak{c}(x)= \sum_{i=0}^{n-1}\textfrak{c}_i x^i \in \C,$ then its reverse $\textfrak{c}^R(x)=  \sum_{i=0}^{n-1}\textfrak{c}^r_{n-1-i} x^i \in \C.$   As $\C$ is a left $\R$-submodule and $( 1+x+\cdots+x^{n-1}) \in \C$,  $( 1+x+\cdots+x^{n-1}) + 3\textfrak{c}^R(x) \in \C.$  \\
    Now,
    \begin{align*}
            ( 1+x+\cdots+x^{n-1}) + 3\textfrak{c}^R(x)= & \sum_{i=0}^{n-1}x^i +  \sum_{i=0}^{n-1}3 \textfrak{c}^r_{n-1-i} x^i\\
        =& \sum_{i=0}^{n-1} \Big( 1 +3 \textfrak{c}^r_{n-1-i}\Big) x^i
        =\sum_{i=0}^{n-1} \textfrak{c}^{rc}_{n-1-i} x^i\\
        =&\textfrak{c}^{RC}(x).
    \end{align*}
     Thus, $ \textfrak{c}^{RC}(x) \in \C$ and hence the code $\C$ is an RC-code.\\
    Conversely, suppose $\C$ is an RC-code. Thus, if  $\textfrak{c}(x)= \sum_{i=0}^{n-1}\textfrak{c}_i x^i \in \C,$  we get $\textfrak{c}^{RC}(x) = \sum_{i=0}^{n-1} \textfrak{c}^{rc}_{n-1-i} x^i  \in \C.$ Since the all-zero codeword $\boldsymbol{0} \in \C$, $\boldsymbol{0}^{RC}= ( 1+x+\cdots+x^{n-1}) \in \C.$\\ Again, $ ( 1+x+\cdots+x^{n-1}) + 3\textfrak{c}^{RC}(x) \in \C$ and
\begin{align*}
     \sum_{i=0}^{n-1}x^i +  \sum_{i=0}^{n-1} 3\textfrak{c}^{rc}_{n-1-i} x^i =& \sum_{i=0}^{n-1} \Big( 1 + 3 \textfrak{c}^{rc}_{n-1-i}\Big) x^i \\
     = & \sum_{i=0}^{n-1} \textfrak{c}^r_{n-1-i} x^i \\
     = & \textfrak{c}^{R}(x).
\end{align*}
Thus, the code $\C$ is an R-code.
\end{proof}

\noindent Now, we introduce a new construction to generate DNA codes from the codes obtained by Theorems \ref{CREven} and \ref{CROdd}. This approach uses the sum of codes to extend an R-code onto DNA code. As linear codes of length $n$ over $\R$  are characterized as submodules of $\R^n$ over $\R$, the sum of two linear codes of the same length over a ring is defined as the sum of two submodules of a module. Clearly, the sum of two linear codes is again a linear code of the same length over the same ring. \\
\begin{definition}
    Let $\C_1$ and $\C_2$ be two linear codes over $\R$ of the same length. The set
    $$\C_1+\C_2 = \{ \textfrak{c}_1+\textfrak{c}_2 ~|~ \textfrak{c}_1 \in \C_1, \textfrak{c}_2 \in \C_2 \}$$
    is a linear code over $\R$ and is called the sum of codes $\C_1$ and $\C_2$. It is the smallest $\R$-submodule of $\R^n$ which contains both the submodules $\C_1$ and $\C_2$ of $\R$-module $\R^n$.
    \end{definition}
\begin{lemma}\label{revsum}
Let $\C_i$ be $(\text{\textbaro },\mathfrak{d},\gamma_i)$-CC codes of length $n $ over $\R,$ for $i=1,2.$ If $\Phi(\C_1)$ and $\Phi(\C_2)$ are R-codes, then $\Phi(\C_1)+\Phi(\C_2)$ is an R-code of length $2n$.
\end{lemma}
\begin{proof}
Let $\C_1$ and $\C_2$ be $(\text{\textbaro },\mathfrak{d},\gamma_1)$-CC and $(\text{\textbaro },\mathfrak{d},\gamma_2)$-CC codes, respectively over $\R$ such that $\Phi(\C_1)$ and $\Phi(\C_2)$ are R-codes of length $2n$.
 Suppose $\textfrak{c}_1=(\textfrak{c}_{10},\textfrak{c}_{11},\dots,\textfrak{c}_{1(n-1)})$ and $\textfrak{c}_2=(\textfrak{c}_{20},\textfrak{c}_{21},\dots,\textfrak{c}_{2(n-1)})$, where $\textfrak{c}_{ij}=a_{ij}+\omega b_{ij},$ $i=1,2$ and $0 \leq j \leq n-1$ be two codewords in $\C_1$ and $\C_2$, respectively. Then
$\Phi(\textfrak{c}_1)+\Phi(\textfrak{c}_2) \in \Phi(\C_1)+\Phi(\C_2)$ for $\Phi(\textfrak{c}_1) \in \Phi(\C_1)$ and $\Phi(\textfrak{c}_2) \in \Phi(\C_2)$. Now we see
\begin{align*}
  \Phi(\textfrak{c}_1)+\Phi(\textfrak{c}_2)=& \big(\phi(\textfrak{c}_{10}),\phi(\textfrak{c}_{11}),\dots,\phi(\textfrak{c}_{1(n-1)})\big) + \big(\phi(\textfrak{c}_{20}),\phi(\textfrak{c}_{21}),\dots,\phi(\textfrak{c}_{2(n-1)})\big)  \\
  = &\big(a_{10},a_{10}+b_{10},a_{11},a_{11}+b_{11}, \dots, a_{1(n-1)}, a_{1(n-1)}+b_{1(n-1)}\big)\\
 & +\big(a_{20},a_{20}+b_{20},a_{21},a_{21}+b_{21}, \dots, a_{2(n-1)}, a_{2(n-1)}+b_{2(n-1)}\big)\\
 = &\big(a_{10}+a_{20},(a_{10}+b_{10})+(a_{20}+b_{20}),a_{11}+a_{21},(a_{11}+b_{11}) \\& +(a_{21}+b_{21}),
 \dots, a_{1(n-1)}+a_{2(n-1)}, (a_{1(n-1)}+b_{1(n-1)}) \\&  +(a_{2(n-1)}+b_{2(n-1)})\big).
\end{align*}
 The reverse of this codeword over $\ziv$ is
\begin{align*}
    (\Phi(\textfrak{c}_1)+\Phi(\textfrak{c}_2))^{R} =& \big((a_{1(n-1)}+b_{1(n-1)})+(a_{2(n-1)}+b_{2(n-1)}), a_{1(n-1)}+a_{2(n-1)}, \\
    &  \dots,(a_{11}+b_{11})+(a_{21}+b_{21}), a_{11}+a_{21}, (a_{10}+b_{10}) \\
    &  +(a_{20}+b_{20}), a_{10}+a_{20}\big)\\
    =& \big(a_{1(n-1)}+b_{1(n-1)},a_{1(n-1)}, \dots, a_{11}+b_{11}, a_{11}, a_{10}+b_{10}, a_{10}\big)\\
    & + \big(a_{2(n-1)}+b_{2(n-1)},a_{2(n-1)}, \dots, a_{21}+b_{21}, a_{21}, a_{20}+b_{20}, a_{20}\big)\\
    =&  \Phi(\textfrak{c}_1)^{R}+\Phi(\textfrak{c}_2)^{R}.
\end{align*}
Since $\Phi(\C_1)$ and $\Phi(\C_2)$ are R-codes, $\Phi(\textfrak{c}_1)^{R} \in \Phi(\C_1)$ and $\Phi(\textfrak{c}_2)^{R} \in \Phi(\C_2)$. Thus $(\Phi(\textfrak{c}_1)+\Phi(\textfrak{c}_2))^{R} \in \Phi(\C_1)+\Phi(\C_2)$ and hence $\Phi(\C_1)+\Phi(\C_2)$ is an R-code.
\end{proof}

\begin{theorem}\label{DNA2}\textbf{(Construction 2)}
    Let $\C_1$ be a $\tdl$-CC code over $\R$ of length $n$ such that $\Phi(\C_1)$ is an R-code and $\C_2$ be the code generated by the all $1$-vector over $\R$ of length $n$. Then the code $\textfrak{D}=\Phi(\C_1)+\Phi(\C_2)$ is a DNA code of length $2n$.
\end{theorem}
\begin{proof}
Let $\Phi(\C_1) $ be an R-code of length $2n$ where $\C_1$ is a $\tdl$-CC code over $\R$ of length $n$.
Let $\C_2$ be the code generated by the all $1$-vector over $\R$ of length $n$, i.e., $\C_2 = \la \textfrak{g}(x) \ra = \la 1+x+\cdots+x^{n-1} \ra.$
As $\textfrak{g}(x)= g_1(x)+\omega g_2(x)$ where $g_1(x)= 1+x+\cdots+x^{n-1}$ and $g_2(x)= 0$, it follows that $g_1(x) $ is self-reciprocal, also $g_2(x)$ satisfies the conditions \textit{2.} and \textit{3.} of Theorem \ref{CREven} or \ref{CROdd} depending on whether $n$ is odd or even, respectively. Thus, $\Phi(\C_2)$
is an R-code of length $2n$. By Lemma \ref{revsum}, $\textfrak{D}=\Phi(\C_1)+\Phi(\C_2)$ is an R-code of length $2n$. \\
   Since $1+x+\cdots+x^{n-1} \in \C_2,$ the all $1$-vector of length $2n$ is in $\Phi(\C_2)$, thus in the code $\textfrak{D}$. Hence, by Theorem \ref{DNA}, $\textfrak{D}$ is a DNA code.
\end{proof}

\section{Computational Results}
In this section, we provide a few examples and codes over $\mathbb{Z}_4$ and their application in DNA codes. Additionally, we compare our obtained codes with the $\mathbb{Z}_4$ database given at \cite{z4codes}. We provide the generator matrices of all the codes obtained in this article at \cite{Z4codes_Gmats}.

\begin{example}
 Let $\C$ be the $(\text{\textbaro }, \mathfrak{d})$-cyclic code of length $n=14$ over $\textfrak{R}$ generated by $\textfrak{g}(x)=x^{11} + \omega x^{10} + 2x^9 + (2\omega  + 1)x^8 + x^6 + x^5 + (\omega  + 3)x^4+ 3x^3 + (3\omega  +2)x^2 + x + \omega  + 1,$ a right divisor of $x^{14}+3$ in $\R[x;\text{\textbaro },\mathfrak{d}]$, where $\mathfrak{d}(\textfrak{r}) =(1+2\omega) (\text{\textbaro }(\textfrak{r})-\textfrak{r})$. In terms of $\textfrak{g}(x)$, the factorization of $x^{14}+3$ is presented by
 \begin{align*}
     x^{14} + 3=& \big(x^3 + (\omega + 3)x^2 + x + \omega + 2\big)\big(x^{11} + \omega x^{10} + 2x^9 + (2\omega  + 1)x^8 \\& + x^6 + x^5 + (\omega  + 3)x^4+ 3x^3 + (3\omega  +2)x^2 + x + \omega  + 1\big).
 \end{align*}
Now, consider the Gray map $\phi(a+\omega b)=(a+b,b)$. Under this map, $\Phi(\C)$ is transformed into a $(\text{\textbaro }, \mathfrak{d})$-cyclic code with parameters $(28, 4^9 , 12)$. This is better than the linear code $(28, 4^9, 9)$ over $\mathbb{Z}_4$ \cite{z4codes}.
\end{example}

 \renewcommand{\arraystretch}{2.0}
\begin{table}[ht]
\begin{center}
\scriptsize
\begin{tabular}{|ccccc|}
\hline  \hline
 $\alpha(\text{\textbaro }(a)-a)$ & $G_M$&$\langle \textfrak{g}(x)\rangle$&  $\C_M$&Comparison/ \\
 &&&&Remark\cite{z4codes}\\
  \hline  \hline

  $1+2\omega $&$~N_1~$&$ x^3 + (3\omega + 3)x^2 + (\omega + 3)x + 2\omega + 1$&$ (16, 4^{12}, 4)$ & Optimal  \\

\hline
  
$1+2\omega $&$~N_1~$&$x^6 + 2x^5 + 2x^4 + 2x^3 + 2x^2 + (3\omega + 2)x + \omega + 3 $&$ (24, 4^{17} 2^1, 4)$ & $ (24, 4^{17}, 4)$  \\

\hline

$1+2\omega $&$~N_1~$&$x^{11} + \omega x^{10} + 2x^9 + (2\omega  + 1)x^8 + x^6 + x^5 + (\omega  + 3)x^4 $&$ (28, 4^9 , 12)$ & $(28, 4^9 , 9)$  \\

&&$+ 3x^3 + (3\omega  +2)x^2 + x + \omega  + 1$&&\\
\hline

$1+2\omega $&$~N_2~$&$x^6 + 3\omega x^5 + (3\omega + 3)x^4 + (3\omega + 2)x^3 + (\omega + 1)x^2 + 1$&$ (32, 4^{20}, 4)$ & Optimal  \\
  \hline
  
$2 $&$~N_2~$&$ x^{10} + (3\omega + 2)x^8 + (3\omega + 3)x^7 + 2x^6 + (2\omega + 3)x^5 $&$ (36, 4^{16}, 8)$ & Optimal  \\
    $ $&$ $&$  +
  (2\omega + 2)x^4  + (3\omega + 2)x^3 + (3\omega + 2)x^2 + 3$&$ $ &   \\
  \hline

$1+2\omega $&$~N_1~$&$x^6 + x^5 + 2x^4 + (2\omega  + 3)x^3 + 2x^2 + (3\omega  + 1)x + \omega $&$ (40, 4^{33}, 4)$ & $(40, 4^{33}, 3)$  \\
  \hline

$1+2\omega $&$~N_1~$&$x^6 + (2\omega  + 2)x^5 + 2x^4 + 2x^2 + (\omega  + 1)x + \omega  + 2$&$ (48, 4^{41} 2^1, 4)$ & $(48, 4^{41}, 3)$  \\
  \hline

$1+2\omega $&$~N_1~$&$x^6 + (3\omega  + 1)x^5 + 2x^4 + (3\omega  + 2)x^3 + 2x^2 + x + 1$&$ (56, 4^{48}, 4)$ & $(56, 4^{48}, 2)$  \\
  \hline
 
$1+2\omega $&$~N_1~$&$x^6 + (2\omega  + 3)x^5 + 2x^4 + (2\omega  + 3)x^3 + 2x^2 + (3\omega  + 3)x + 3\omega  + 2$&$ (80, 4^{73} 2^1, 4)$ & $ (80, 4^{73},1)$  \\
  \hline

$1+2\omega $&$~N_1~$&$x^6 + (\omega  + 1)x^5 + 2x^4 + x^3 + 2x^2 + (2\omega  + 1)x + \omega  + 2$&$ (120, 4^{113} 2^1, 4)$ & $ (120, 4^{113} , 1)$  \\
  \hline

\end{tabular}
\caption{$(\text{\textbaro },\mathfrak{d})$-cyclic codes over $\R$ of length $n$ with Gray images}
\label{tdcyclic}
\end{center}
\end{table}

\begin{example}
Let $\C = \la \textfrak{g}(x) \ra$ be the $(\text{\textbaro }, \mathfrak{d})$-cyclic code over $\R$ of length $14$ such that 
\begin{itemize}
    \item  $\textfrak{g}(x) = x^8 + (2\omega + 3)x^7 + 2x^5 + 2x^4 + 2x^3 + (2\omega + 3)x + 1$.
\item $\mathfrak{d}(\textfrak{r}) = 2(\text{\textbaro }(\textfrak{r}) - \textfrak{r})$ for all $\textfrak{r} \in \R$.
\end{itemize}
 Then $\textfrak{g}(x)|_rx^{14} - 3$ in $\R[x; \text{\textbaro }, \mathfrak{d}]$.
Now, consider the Gray map $\Phi(a + \omega b) = (a, a+b)$. Under this map, $\Phi(\C)$ is transformed into a $(\text{\textbaro }, \mathfrak{d},3)$-cyclic code with parameters $(28, 4^{12} , 8)$. Furthermore, it satisfies the conditions of Theorem $\ref{CREven}$, which implies that it is an R-code over $\mathbb{Z}_4$.
\end{example}

\begin{landscape}
 \renewcommand{\arraystretch}{2.0}
    \begin{table}
\begin{center}
\begin{tabular}{|cccccc|}
	\hline\hline
 $\alpha(\text{\textbaro }(a)-a)$ &$\gamma$& $G_M$&$\langle \textfrak{g}(x)\rangle$&  $\C_M$&Comparison/ Remark\cite{z4codes}\\
  \hline
  \hline

  $1+2\omega $&$3$&$~N_1~$&$x^3 + \omega x^2 + (3\omega  + 2)x + \omega  + 1$&$ (12, 4^7 2^2, 4)$ & $(12, 4^7 , 4)$\\
  \hline

$1+2\omega $&$1$&$~N_3~$&$x^4 + (2\omega  + 3)x^3 + (3\omega  + 3)x^2 + (2\omega  + 3)x + \omega  + 2$&$ (16, 4^{10} 2^1, 4)$ & $(16, 4^{10} , 4)$ \\
  \hline
  
$1+2\omega $&$3$&$~N_1~$&$x^6 + (2\omega  + 3)x^5 + (2\omega  + 2)x^4 + 2\omega x^3 + (2\omega  + 2)x^2 + x + 1$&$ (20, 4^8 2^4, 8)$ & $(20, 4^8, 8)$ \\
  \hline

$1+2\omega $&$1$&$~N_1~$&$x^5 + (3\omega  + 2)x^4 + x^3 + (3\omega  + 2)x^2 + 3x + 3$&$ (24, 4^{18} 2^1, 4)$ & $(24,4^{18},4)$ \\
  \hline

$1+2\omega $&$3$&$~N_2~$&$x^8 + (2\omega + 3)x^7 + 2x^5 + 2x^4 + 2x^3 + (2\omega + 3)x + 1 $&$ (28, 4^{12} , 8)$ & Optimal  \\

\hline

  $1+2\omega $&$3$&$~N_1~$&$x^3 + 2\omega x^2 + (2\omega  + 1)x + 2\omega  + 1$&$ (28, 4^{22} 2^3, 3)$ & New Code\\
  \hline

  $1+2\omega $&$1$&$~N_1~$&$x^5 + (3\omega  + 2)x^4 + 2x^3 + 3\omega x^2 + (\omega  + 2)x + \omega  + 3$&$ (32, 4^{26} 2^1, 4)$ & $ (32, 4^{26}, 3)$\\
  \hline

  $1+2\omega $&$1$&$~N_1~$&$x^5 + (3\omega  + 2)x^4 + 3\omega x^3 + 3\omega x^2 + (3\omega  + 2)x + \omega  + 1$&$ (40, 4^{31} 2^4, 4)$ & $ (40, 4^{31} , 3)$\\
  \hline

  $1+2\omega $&$1$&$~N_1~$&$x^5 + (3\omega  + 2)x^4 + x^3 + (\omega  + 2)x^2 + (2\omega  + 3)x + 3$&$ (48, 4^{42} 2^1, 4)$ & $ (48, 4^{42} , 3)$\\
  \hline

  $1+2\omega $&$1$&$~N_1~$&$x^3 + (2\omega  + 1)x^2 + 1$&$ (56, 4^{50} 2^3, 3)$ & $ (56, 4^{50} , 1)$\\
  \hline

  $1+2\omega $&$1$&$~N_1~$&$x^5 + (3\omega  + 2)x^4 + (\omega  + 1)x^3 + 3\omega x^2 + (2\omega  + 1)x + \omega  + 3$&$ (56, 4^{50} 2^1, 4)$ & $ (56, 4^{50}, 1)$\\
  \hline

  $1+2\omega $&$1$&$~N_1~$&$x^5 + (3\omega  + 2)x^4 + (\omega  + 2)x^3 + 2x^2 + (3\omega  + 2)x + 3\omega  + 3$&$ (60, 4^{54} 2^1, 4)$ & $ (60, 4^{54} , 2)$\\
  \hline

  $1+2\omega $&$1$&$~N_1~$&$x^5 + (3\omega  + 2)x^4 + 2x^3 + 3\omega x^2 + \omega x + \omega  + 1$&$ (64, 4^{58} 2^1, 4)$ & $ (64, 4^{58} , 1)$\\
  \hline

  $1+2\omega $&$3$&$~N_1~$&$x^6 + (3\omega + 3)x^5 + 2x^4 + 2x^2 + (3\omega + 3)x + \omega$&$ (84, 4^{76} 2^2, 4)$ & $ (84, 4^{76}, 1)$\\
  \hline

\end{tabular}
\caption{$\tdl$-CC codes over $\R$ with Gray images}
\label{CC_Tbl2}
\end{center}
\end{table}
\end{landscape}

\begin{example}\label{DNA codes}
Let $\C=\la g(x)\ra = \la x^{10} + 3x^9 + 2x^7 + 2x^5 + 2x^4 + 2x^3 + 3x + 1 \ra$ be the $(\text{\textbaro }, \mathfrak{d})$-cyclic code  over $\R$ of length $n=18$, where $\textfrak{g}(x) |_r x^{18}-1$ in $\R[x;\text{\textbaro },\mathfrak{d}]$ and $\alpha =2$.
Now consider the Gray map $\phi(a+\omega b)=(a,a+b)$. Under this map, $\Phi(\C)$ is transformed into a $(\text{\textbaro }, \mathfrak{d})$-cyclic code with parameters $(36, 4^{16}, 8)$. Further, it satisfies the conditions of Theorems \ref{CREven} and \ref{DNA}. Hence, it is a DNA code over $\mathbb{Z}_4$.
\end{example}

\renewcommand{\arraystretch}{1.53}
\begin{table}[ht]
\begin{center}
\scriptsize
\begin{tabular}{|cccc|}
	\hline
    \hline
$~~n~~$& $\alpha(\text{\textbaro }(a)-a)$&$\langle \textfrak{g}(x)\rangle$&  $\Phi(\C)$ \\
 \hline
 \hline
  $8 $&$ 1+2\omega $ &$x^3 + 3x^2 + 3x + 1$&$(16, 4^{10}, 4) $  \\
  \hline
  $ 10 $&$ 3+2\omega $ & $x^8 + x^6 + x^4 + x^2 + 1 $ & $(20, 4^4, 5) $   \\
 \hline
  $12 $&$  2$ &$x^4 + 3x^3 + 2x^2 + 3x + 1 $&$ (24, 4^{16}, 4)$  \\
  \hline
   $14 $&$ 1+2\omega $ &$ x^8 + 3x^7 + 2x^6 + 2x^4 + 2x^3 + 2x^2 + 3x + 1$&$(28, 4^{12}, 8) $  \\
  \hline
  $15 $&$ 2
 $ &$ x^6 + 2x^5 + 3x^4 + 3x^3 + 3x^2 + 2x + 1$&$ (30, 4^{18}, 4)$  \\
  \hline
  $16 $&$ 1+2\omega $ &$x^{11} + x^{10} + 3x^9 + 3x^8 + 2x^7 + 2x^6 + 2x^5 + 2x^4 + 3x^3 + 3x^2 + x
    + 1  $&$ (32, 4^{10}, 8)$  \\
  \hline
  $ 16$&$1 + 2\omega  $ &$x^6 + 2x^5 + 3x^4 + 2x^3 + 3x^2 + 2x + 1 $&$ (32, 4^{20}, 4)$  \\
  \hline
  $18 $&$ 2 $ &$x^6 + 2x^5 + 2x^4 + 3x^3 + 2x^2 + 2x + 1  $&$(36, 4^{24}, 3) $  \\
  \hline
  $18 $&$ 2 $ &$ x^{10} + 3x^9 + 2x^7 + 2x^5 + 2x^4 + 2x^3 + 3x + 1 
$&$(36, 4^{16}, 8) $  \\
  \hline
  $20 $&$ 2 $ &$x^4 + (2\omega + 1)x^3 + (2\omega + 3)x^2 + (2\omega + 1)x + 1
 $&$ (40, 4^{32}, 4)$  \\
   \hline
  $30 $&$2  $ &$ x^6 + 2x^5 + x^4 + 3x^3 + x^2 + 2x + 1$&$(60, 4^{48}, 4) $  \\
  \hline
 
$30 $&$ 1+2\omega $ &$ x^{28} + x^{26} + x^{24} + x^{22} + x^{20} + x^{18} + x^{16} + x^{14} + x^{12} + x^{10} + x^8 + x^6 + x^4 + x^2 + 1
$&$(60, 4^4, 15) $  \\
  \hline
  
\end{tabular}
\caption{DNA codes from $(\text{\textbaro },\mathfrak{d})$-cyclic codes}
\label{Tbl2}
\end{center}
\end{table}
\vspace{-0.3 cm}

\begin{table}
\begin{center}
\tiny
\begin{tabular}{|cc|}
						\hline
$CTCGCTCGCTCGCTCGCTCGCTCGCTCGCTCGCTCGCT
CGCT$&
$TCTATCTATCTATCTATCTATCTATCTATCTATCTATC
TATC$\\
$CGCTCGCTCGCTCGCTCGCTCGCTCGCTCGCTCGCTCG
CTCG$&
$GATAGATAGATAGATAGATAGATAGATAGATAGATAGA
TAGA$\\
$AACCAACCAACCAACCAACCAACCAACCAACCAACCAA
CCAA$&
$TTGGTTGGTTGGTTGGTTGGTTGGTTGGTTGGTTGGTT
GGTT$\\
$ACCAACCAACCAACCAACCAACCAACCAACCAACCAAC
CAAC$&
$TGGTTGGTTGGTTGGTTGGTTGGTTGGTTGGTTGGTTG
GTTG$\\
$CACACACACACACACACACACACACACACACACACACA
CACA$&
$GAGCGAGCGAGCGAGCGAGCGAGCGAGCGAGCGAGCGA
GCGA$\\
$ATAGATAGATAGATAGATAGATAGATAGATAGATAGAT
AGAT$&
$TAGATAGATAGATAGATAGATAGATAGATAGATAGATA
GATA$\\
$ACACACACACACACACACACACACACACACACACACAC
ACAC$&
$CTATCTATCTATCTATCTATCTATCTATCTATCTATCT
ATCT$\\
$GCTCGCTCGCTCGCTCGCTCGCTCGCTCGCTCGCTCGC
TCGC$&
$TTTTTTTTTTTTTTTTTTTTTTTTTTTTTTTTTTTTTT
TTTT$\\
$CCCCCCCCCCCCCCCCCCCCCCCCCCCCCCCCCCCCCC
CCCC$&
$GGGGGGGGGGGGGGGGGGGGGGGGGGGGGGGGGGGGGG
GGGG$\\
$ATCTATCTATCTATCTATCTATCTATCTATCTATCTAT
CTAT$&
$TCGCTCGCTCGCTCGCTCGCTCGCTCGCTCGCTCGCTC
GCTC$\\
$GTGTGTGTGTGTGTGTGTGTGTGTGTGTGTGTGTGTGT
GTGT$&
$CGAGCGAGCGAGCGAGCGAGCGAGCGAGCGAGCGAGCG
AGCG$\\
$GCGAGCGAGCGAGCGAGCGAGCGAGCGAGCGAGCGAGC
GAGC$&
$TGTGTGTGTGTGTGTGTGTGTGTGTGTGTGTGTGTGTG
TGTG$\\
$AGATAGATAGATAGATAGATAGATAGATAGATAGATAG
ATAG$&
$CAACCAACCAACCAACCAACCAACCAACCAACCAACCA
ACCA$\\
$GTTGGTTGGTTGGTTGGTTGGTTGGTTGGTTGGTTGGT
TGGT$&
$CCAACCAACCAACCAACCAACCAACCAACCAACCAACC
AACC$\\
$GGTTGGTTGGTTGGTTGGTTGGTTGGTTGGTTGGTTGG
TTGG$&
$AGCGAGCGAGCGAGCGAGCGAGCGAGCGAGCGAGCGAG
CGAG$\\
$TATCTATCTATCTATCTATCTATCTATCTATCTATCTA
TCTA$&
$AAAAAAAAAAAAAAAAAAAAAAAAAAAAAAAAAAAAAA
AAAA$\\
			\hline			
		\end{tabular}
  \caption{DNA codewords of length $21$ corresponding to Example \ref{DNA construction}}		
  \label{codewords2}
   \end{center}
\end{table}

\begin{example}\label{DNA construction}
Let $\C_1 =\la g(x)\ra = \la x^{20} + 3x^{19} + x^{18} + 3x^{17} + x^{16} + 3x^{15} + x^{14} + 3x^{13} + x^{12} + 3x^{11} +  x^{10} + 3x^9 + x^8 + 3x^7 + x^6 + 3x^5 + x^4 + 3x^3 + x^2 + 3x + 1  \ra$ be the $(\text{\textbaro }, \mathfrak{d}, 3)$-CC code over $\R$ of length $21$, where $\textfrak{g}(x) |_r x^{21}-1$ in $\R[x;\text{\textbaro },\mathfrak{d}]$ and $\mathfrak{d}(\textfrak{r}) =(1+2\omega) (\text{\textbaro }(\textfrak{r})-\textfrak{r})$.\\
Now under the Gray map $\phi(a+\omega b)=(a,a+b)$, $\Phi(\C_1)$ is a $(\text{\textbaro }, \mathfrak{d},3)$-CC code with the parameters $(42, 4^2, 21)$. Also, it meets all the conditions of Theorem \ref{CROdd}, therefore it is an R-code over $\mathbb{Z}_4$. 
Let $\C_2$ be the code as in Theorem \ref{DNA2} with length $42$ over $\R$. Then $\Phi(\C_1)+\Phi(\C_2)$ is a DNA code with parameters $(42, 4^2 2^1, 21)$. The corresponding DNA codewords are presented in Table \ref{codewords2}.
\end{example}

\begin{table}[ht]
\begin{center}
\scriptsize
    \begin{tabular}{|cccccc|}
       \hline
       \hline
        $~~n~~$ & $~~\gamma~~$ & $\alpha $ & $\la \textfrak{g}(x) \ra$ & $\Phi(\textfrak{C})$ & DNA codes \\
\hline
\hline
$6$ & $1$ &$ 2$ & $x^3 + 3x^2 + 3x + 1$ & $(12, 4^6, 4)$ & $(12, 4^7, 4)$ \\
\hline
$8$ & $ 1$ &$ 2$& $x^4 + 2 \omega x^3 + (2 \omega + 2)x^2 + 2 \omega x + 1$ & $(16, 4^8, 4)$ & $ (16, 4^8 2^1, 4)$ \\
\hline
$8$ & $ 1$ &$ 1+2\omega$& $x^6 + 2x^5 + x^4 + x^2 + 2x + 1$ & $(16, 4^4 , 8)$ & $ (16, 4^4 2^1, 8)$ \\
\hline

$10$ & $3$ &$ 2$& $x^4 + (2\omega  + 1)x^3 + 3x^2 + (2\omega  + 1)x + 1$ & $(20, 4^{12} , 4)$ & $(20, 4^{12} 2^1, 4)$ \\
\hline

$12 $ & $1 $ &$2 $& $ x^{10} + 2x^9 + x^8 + x^6 + 2x^5 + x^4 + x^2 + 2x + 1$ & $(24, 4^4 , 12) $ & $(24, 4^4 2^1, 12) $ \\
\hline
 
$12 $ & $1 $ &$2 $& $ x^4 + (2\omega + 3)x^3 + (2\omega + 2)x^2 + (2\omega + 3)x + 1$ & $(24, 4^{16}, 4) $ & $(24, 4^{16} 2^1, 4) $ \\

\hline
$ 14$ & $3 $ &$2 $& $x^8 + (2\omega + 3)x^7 + 2x^5 + 2x^4 + 2x^3 + (2\omega + 3)x + 1 $ & $(28, 4^{12} , 8) $ & $(28, 4^{13}, 8) $ \\
\hline
$14 $ & $3 $ &$ 2$& $ x^6 + (2\omega + 1)x^5 + 3x^4 + (2\omega + 1)x^3 + 3x^2 + (2\omega + 1)x + 1$ & $ (28, 4^{16} , 4)$ & $(28, 4^{16} 2^1, 4) $ \\
\hline
$ 15$ & $3 $ &$2 $& $x^9 + 2x^8 + x^7 + 3x^6 + 2x^5 + 2x^4 + 3x^3 + x^2 + 2x + 1 $ & $ (30, 4^{12} , 6)$ & $ (30, 4^{13} , 6)$ \\
\hline
$ 16$ & $1 $ &$ 2$& $x^{10} + 2x^9 + 3x^8 + 2x^7 + 2x^6 + 2x^4 + 2x^3 + 3x^2 + 2x + 1 $ & $(32, 4^{12} , 8) $ & $(32, 4^{12} 2^1, 8) $ \\
\hline
$20 $ & $ 1 $& $ 1+2\omega $ & $
x^{16} + 3x^{15} + 2x^{13} + 2x^{12} + 3x^{11} + x^{10} + x^6 + 3x^5  $ & $(40, 4^8 , 12) $ & $(40, 4^9 , 12) $ \\

$ $ & $ $&$ $ & $ + 2x^3 + 2x^2 + 3x + 1 $ & $ $ & $ $ \\
\hline
$20 $ & $ 1$ &$ 2$& $ x^{12} + 2x^{11} + 3x^{10} + 2x^9 + 2x^6 + 2x^4 + 2x^3 + 3x^2 + 2x + 1 $ & $(40, 4^{16} , 8) $ & $ (40, 4^{17} , 8)$ \\
\hline

$21 $ & $ 3 $& $ 1+2\omega $ & $
x^{20} + 3x^{19} + x^{18} + 3x^{17} + x^{16} + 3x^{15} + x^{14} + 3x^{13} + x^{12} + 3x^{11} $ & $(42, 4^2, 21) $ & $(42, 4^2 2^1, 21) $ \\
$ $ & $ $&$ $ & $ +  x^{10} + 3x^9 + x^8 + 3x^7 + x^6 + 3x^5 + x^4 + 3x^3 + x^2 + 3x + 1  $ & $ $ & $ $ \\
\hline

$23 $ & $ 3 $& $ 3+2\omega $ & $ x^{22} + 3x^{21} + x^{20} + 3x^{19} + x^{18} + 3x^{17} + x^{16} + 3x^{15} + x^{14} + 3x^{13} + x^{12} $ & $(46, 4^2, 23)$ & $(46, 4^2 2^1, 23)$ \\
$ $ & $ $&$ $ & $  + 3x^{11} +  x^{10} + 3x^9 + x^8 + 3x^7 + x^6 + 3x^5 + x^4 + 3x^3 + x^2 + 3x + 1  $ & $ $ & $ $ \\
\hline

 $24 $ & $1 $ &$ 1+2\omega$& $x^{22} + 2x^{21} + x^{20}  + x^{18} + 2x^{17}+  x^{16} + x^{14} + 2x^{13} 
 $ & $(48, 4^4, 24) $ & $(48, 4^4 2^1, 24) $ \\
$ $ & $ $&$ $ & $ + x^{12} + x^{10} + 2x^9 + x^8 + x^6 + 2x^5+ x^4 + x^2 + 2x + 1 $ & $ $ & $ $ \\

 \hline
 
$28 $ & $1 $&$ 1+2\omega$ & $x^{26} + 2x^{25} + x^{24} + x^{22} + 2x^{21} + x^{20} + x^{18} + 2x^{17} + x^{16} + x^{14}  $ & $(56, 4^4, 28) $ & $(56, 4^4 2^1, 28) $ \\
$ $ & $ $&$ $ & $ + 2x^{13} + x^{12} + x^{10} + 2x^9 + x^8 + x^6 + 2x^5 + x^4 + x^2 + 2x + 1  $ & $ $ & $ $ \\

 \hline
$30 $ & $3$ &$ 2$& $x^6 + 2x^5 + 3x^4 + (2\omega + 3)x^3 + 3x^2 + 2x + 1$ & $(60, 4^{48}, 4)$ & $(60, 4^{48} 2^1, 4)$ \\
\hline

\end{tabular}
   \caption{Reversible ($\Phi(\C)$) and DNA codes from $\tdl$-cyclic codes}
    \label{Const.tab}
\end{center}
\end{table}

\begin{remark}
In all the code tables, we write $\alpha$ in place of $\alpha(\text{\textbaro }(\textfrak{r})-\textfrak{r}),~ \textfrak{r} \in \R$ for the value of $\mathfrak{d}(\textfrak{r})$. Further, we compare our codes over $\mathbb{Z}_4$ by the database \cite{z4codes} with respect to Lee distances.
\end{remark}

\section{Conclusion}
In this work, we have explored $\tdl$-CC codes over the ring $\R$. We have determined the criteria for these codes to be reversible and reversible-complement. Further, we have shown that DNA codes can be constructed by extending the obtained R-codes over $\R$. Here, we proposed an efficient method to obtain classical codes with the help of a generator matrix. Implementing the proposed constructions, we have produced many better and optimal linear codes over $\mathbb{Z}_4$ along with several DNA codes. These codes are presented at \cite{Z4codes_Gmats} with their corresponding generator matrices. It would be valuable to examine the experimental performance of these codes. Additionally, exploring different construction techniques to derive DNA codes from generalized skew cyclic codes remains an open problem for future research.
	
\section*{Acknowledgements}
This work is financially supported by the Ministry of Education, Govt. of India, under PMRF-ID: 2702443, and by the Department of Science and Technology, Government of India, under Project CRG/2020/005927, as per Diary No. SERB/F/6780/2020-2021 dated 31 December 2020. The authors express their gratitude to the Indian Institute of Technology Patna for providing the essential research facilities.
\section*{Declarations}
\noindent \textbf{Data Availability Statement}: The authors affirm that the information corroborating the study's conclusions is contained in the paper. The corresponding author can be contacted for any clarification. \\
\textbf{Competing interests}: Regarding the publishing of this paper, the authors affirm that they have no conflicts of interest.\\
\textbf{Use of AI tools declaration:}
The authors affirm that this paper was not produced using AI tools.

\end{document}